\documentclass[10pt]{article}
\usepackage{amsthm}
\usepackage{amsmath,amssymb}
\title{\large \textbf{Conditional coloring of some parameterized graphs}}
\author{
\small P.~Venkata Subba Reddy  and K.~Viswanathan Iyer \\
\small Dept. of Computer Science and Engineering \\  
\small National Institute of Technology \\ 
\small Tiruchirapalli 620 015, India  \\
\small email : venkatpalagiri@gmail.com, kvi@nitt.edu
} 
\date{\small 19 May 2010}    
\begin{document}
\maketitle
\begin{abstract}  
For integers $k>0$ and $r>0$, a conditional $(k,r)$-coloring of a graph $G$ is a proper $k$-coloring of the vertices of $G$ such that every vertex $v$ of degree $d(v)$ in $G$ is adjacent to vertices with at least $\min\{r, d(v)\}$ different colors. The smallest integer $k$ for which a graph $G$ has a conditional $(k,r)$-coloring is called the $r$th order conditional chromatic number, denoted by $\chi_r(G)$. For different values of $r$ we obtain $\chi_r(G)$ of certain parameterized graphs viz., Windmill graph, line graph of Windmill graph, middle graph of Friendship graph, middle graph of a cycle, line graph of Friendship graph, middle graph of complete $k$-partite graph and middle graph of a bipartite graph. 
\end{abstract}
\section{Introduction}
Let $G= (V(G),E(G))$ be a simple, connected, undirected graph. For a vertex $v \in V(G)$, the \textit{neighborhood} of $v$ in $G$ is defined by $N_G(v)$= \{$u \in V(G):(u,v) \in E(G)$\}, and the degree of $v$ is denoted by $d(v)$=$|N_G(v)|$.  For an integer $k>0$, a \textit{proper} $k$-coloring of a graph $G$ is a surjective mapping $c \colon V(G) \to \{1,\ldots,k \}$ such that if $(u,v) \in E(G)$, then $c(u) \neq c(v)$. The smallest $k$ such that $G$ has a proper $k$-coloring is the \textit{chromatic number} $\chi(G)$  of $G$. Given a set $S \subseteq V(G)$ we define $c(S)= \{c(u) : u \in S$ \}. For integers $k>0$ and $r>0$, a \textit{conditional} $(k,r)$-coloring of $G$ is a surjective mapping $c \colon V(G) \to \{1,\ldots,k \}$ such that both the following conditions hold:\begin{quote}
\textit{
(C1) If $(u,v) \in E(G)$, then $c(u) \neq c(v)$. \\
(C2) For any $v \in V(G)$, $|c(N_G(v))| \geq $ min \{$d(v),r$ \}. 
}
\end{quote}
For a given integer $r>0$, the smallest integer $k$ such that $G$ has a conditional $(k,r)$-coloring is called the \textit{$r^{th}$ order conditional chromatic number} of $G$, denoted by $\chi_r(G)$. It is proved in ~\cite{Li2}, that the problem of conditional $(k, r)$-coloring of a graph is hard.

We use the following definitions of certain graphs. The \textit{middle graph } $M(G)$ of $G$  is the graph whose vertex set corresponds to $V(G) \cup E(G)$; in $M(G)$ two vertices are adjacent iff 
$(i)$ they are adjacent edges of $G$ or 
$(ii)$ one is a vertex and the other is an edge incident with it ~\cite{danm}. The \textit{Windmill graph} $Wd(k,n)$ consists of $n$ copies of $K_k$ and identifying one vertex from each $K_k$ as the common center vertex. In particular $Wd(3,n)$ is called the \textit{Friendship graph} $F_n$ ~\cite{Galin}. A complete $k$-partite graph  $K_{n_1,\ldots,n_k}$ has vertex set $V=V_1 \cup \ldots \cup V_k$ where $V_1,\ldots,V_k$ are mutually disjoint with $|V_i|=n_i$; each vertex $v \in V_i$ is connected to all vertices of $V \setminus V_i, \; i=1,\ldots,k$. For undefined notations/terminology see standard texts in graph theory such as ~\cite{west}.
\section{Conditional colorability of some special graphs}
We begin with two lemmas followed by our propositions.
\newtheorem{lem1}{Lemma}
\begin{lem1}
For $r \leq \Delta$ let \textit{Vset-$d2r$} in a graph $G$ be a set $S_{d2r} \subseteq V(G)$ with the following two properties:
\begin{quote}
$(i)$ For all $u \in S_{d2r}$,  $d(u) \leq r$.  \\
$(ii)$ For all $u_1,u_2 \in S_{d2r}$ either $(u_1,u_2) \in E(G)$ or there exists a $u_3 \in S_{d2r}$ such that $u_1,u_2 \in N(u_3)$ or both.
\end{quote}
Then $\chi_r(G) \geq |S_{d2r}|$.
\end{lem1}
\begin{proof}
Assume that $\chi_r(G) < |S_{d2r}|$. Then there exist at least two vertices $u_1,u_2 \in S_{d2r}$ such that $c(u_1)=c(u_2)$. By the definition of Vset-$d2r$    $(ii)$ holds; if $(u_1,u_2) \in E(G)$, (C1) is violated; if $u_3 \in S_{d2r}$ such that $u_1,u_2 \in N(u_3)$, $|c(N(u_3))|< min\{r,d(u_3)\}=d(u_3)$ and hence (C2) is violated at $u_3$. Therefore  $\chi_r(G) \geq |S_{d2r}|$.
\end{proof}

\newtheorem{lem2}[lem1]{Lemma}
\begin{lem2}
Given a graph $G$, let $c \colon V(G) \to \{1,\ldots,k \}$ be a coloring such that for a given $r$, $c$ satisfies (C2). Let the condition (C3) be
\begin{quote}
(C3) For each edge $uv$ in $G$ there exists a vertex $w$ such that $d(w)$ $\leq r$ and $u,v \in N_{G}(w)$.
\end{quote}
If $G$ satisfies (C3) also then $c$ satisfies (C1) and hence $c$ defines a conditional $(k,r)$-coloring of $G$.
\end{lem2}
\begin{proof}
The proof by contradiction is straightforward.
%Assume that (C1) is not satisfied; then there exists two vertices $u,v \in V(G)$ such that $(u,v) \in E(G)$ and $c(u) = c(v)$.  Since for each edge $(u,v) \in E(G)$ there exists a vertex $w$ such that $d(w) \leq r$ and $u,v \in N(w)$ so (C2) is violated at $w$; a contradiction.
\end{proof}
\newtheorem{thm1}{Proposition}  
\begin{thm1} 
For the Windmill graph, we have \[\chi_r(Wd(k,n)) = \left\{ 
\begin{array}{l l}
k, & \quad \mbox{if $2 \leq r \leq k-1$. \textsl{}}\\  
min \{r,\Delta\}+1, & \quad \mbox{if $ r \geq k$. \textsl{}} \\ \end{array} \right. \]
\end{thm1}
\begin{proof}
Every vertex $v$ of $Wd(k,n)$ is contained in a $K_k$; it can be seen that $|c(N(v))| \geq k-1$ in any proper coloring $c$ of $Wd(k,n)$. Therefore if $2 \leq r \leq k-1 $,  conditional $(\chi(Wd(k,n)),r)$-coloring of $Wd(k,n)$ exists and we know that $\chi(Wd(k,n))=k$. 
From ~\cite{Lai1} we have $\chi_r(G) \geq min \{r, \Delta \}+1$. Taking $G=Wd(k,n)$ we get $\chi_r(Wd(k,n))$  $\geq min \{r, \Delta \}+1$. In $Wd(k,n)$ only the center vertex has a degree $n(k-1)> k$. For a $k'$ if $k' > k$, then every $k$-colorable graph is also $k'$-colorable. Hence if $ r \geq k$, then a proper $(min \{r,\Delta\}+1)$-coloring of $Wd(k,n)$ exists, which is also a conditional $(min \{r,\Delta\}+1,r)$-coloring. Therefore $\chi_r(Wd(k,n)) \leq min \{r, \Delta \}+1$. Hence $\chi_r(Wd(k,n)) = min \{r, \Delta \}+1$  if $r \geq k$.    
\end{proof}
\newtheorem{thm}[thm1]{Proposition} 
\begin{thm}
Let $L(Wd(k,n))$ be the line graph of $Wd(k,n)$. Then 
\begin{equation*}
\chi_\Delta(L(Wd(k,n)))= n(k-1)+\binom{k-1}{2}= z \;(say).
\end{equation*} 
\end{thm}
\begin{proof}
It follows that $|V(L(Wd(k,n)))|=n\binom{k}{2} = inx(k,n)$ (say). Let $V(L(Wd(k,n)))= \{v_1,\ldots,v_{inx(k,n)}\}$. We assume that in $Wd(k,n)$, for all $1 \leq i \leq n,   
 v_{(i-1)(k-1)+1}$\  to\  $v_{i(k-1)}$ and $v_{n(k-1)+inx(k-1,i-1)+1}$\  to\  $v_{n(k-1)+inx(k-1,i)}$ represent respectively the edges of $i^{th}$ copy of $K_k$ incident with and not incident with the center vertex. It can be seen that $L(Wd(k,n))$ has a clique $\{ v_1,\ldots,v_{n(k-1)} \}$ and a Vset-$d2r$ $S_{d2r}=\{v_1,\ldots,v_z\}$. By lemma 1, $\chi_\Delta(L(Wd(k,n))) \geq |S_{d2r}|=z$. We now define the coloring assignment $c \colon V(L(Wd(k,n))) \to \{1,\ldots,z \}$ as follows: 
\begin{equation*}
c(v_i) =
\begin{cases}
i, & \text{if $1 \leq i \leq z.$}\\
i \mod {\binom{k-1}{2}}+n(k-1)+1, & \text{otherwise.}
\end{cases}
\end{equation*}
In $c$ the \textit{if}-case uses $z$ and the \textit{otherwise}-case doesn't use any extra color. For all $1 \leq i \leq n$ \\ $|c(\{v_{(i-1)(k-1)+1},\ldots,v_{i(k-1)}\})$ $\cup$  $c(\{v_{n(k-1)+ inx(k-1,i-1)+1},\ldots,v_{n(k-1)+ inx(k-1,i)}\})|$ \\ = $|\{v_{(i-1)(k-1)+1},\ldots,v_{i(k-1)}\} \cup \{v_{n(k-1)+inx(k-1,i-1)+1},\ldots,v_{n(k-1)+inx(k-1,i)}\}|$ and \\ $|c\{v_1,\ldots,v_{n(k-1)}\}|=n(k-1)$; hence (C2) is satisfied at all the vertices. The graph $G= L(Wd(k,n))$ can be seen to satisfy (C3) in lemma 2. By lemma 2, $\chi_\Delta (L(Wd(k,n))) \leq z$; hence $\chi_\Delta (L(Wd(k,n))) = z$.
\end{proof}
%% T or P 3
\newtheorem{thm2}[thm1]{Proposition} 
\begin{thm2}
Let $L(F_n)$ be the line graph of $F_n$. Then
\begin{equation*}
\chi_r(L(F_n)) =
\begin{cases}
2n, & \text{if $r < \Delta$.}\\ 
2n+1 , & \text{if $r = \Delta$.}
\end{cases}
\end{equation*}
\end{thm2}
\begin{proof}
We have the following two cases: \\
\textbf{Case 1:} $r < \Delta:$ Let $r'=\Delta-1$. Since $|V(L(F_n))|=3n$, let $V(L(F_n))= \{v_1,\ldots,v_{3n}\}$. We assume that for all $1 \leq i \leq n, v_{2i-1}$ and $v_{2i}$ represent the edges of $i^{\rm th}$ copy of $K_3$ incident with the center vertex and $v_{2n+i}$ represents the edge of $i^{\rm th}$ copy of $K_3$ not incident with the center vertex of $F_n$. As $\{v_1,\ldots,v_{2n} \}$ is the maximum size clique of $L(F_n)$, $\omega(L(F_n))=2n$. Since $\chi_r(G) \geq \omega(G)$, taking $r=r'$ and $G=L(F_n)$ we have $\chi_{r'}(L(F_n))\geq 2n$. We now define the coloring assignment $c \colon V(L(F_n)) \to \{1,\ldots,2n\}$ as follows: 
\begin{equation*}
c(v_i) =
\begin{cases}
i, & \text{if $1 \leq i \leq 2n$.}\\
2n, & \text{if $2n+1 \leq i \leq 3n-1$.}\\
1, & \text{if $i = 3n$.}
\end{cases}
\end{equation*}
Note that in $c$ the first case uses $2n$ colors and the remaining cases use no new colors. It can be verified that $c$ defines a conditional $(2n,r')$-coloring of $L(F_n)$. Thus $\chi_{r'} (L(F_n)) \leq 2n$; hence $\chi_{r'} (L(F_n))= 2n$. From ~\cite{Lai1} it follows that $\omega(G) \leq \chi_{r_1}(G) \leq \chi_{r_2}(G)$ if $r_1 \leq r_2$. Taking $G=L(F_n),r_1=r$ and $r_2=r'$ it follows that $\chi_r(L(F_n))= 2n$. \\
\textbf{Case 2:} $r = \Delta:$ Since $F_n=Wd(3,n)$, by theorem 1 we have $\chi_\Delta(L(Wd(3,n)))=\chi_\Delta(L(F_n))=2n+1$.
\end{proof}
\newtheorem{thm3}[thm1]{Proposition} 
\begin{thm3}
Let $M(K_{n_1,\ldots,n_k})$ be the middle graph of $K_{n_1,\ldots,n_k}$. Then 
\begin{equation*}
\chi_\Delta(M(K_{n_1,\ldots,n_k}))= k+l.
\end{equation*} 
where $n=\sum_{i=1}^k n_i$ and $l=1/2 \sum_{i=1}^k n_i(n-n_i)$.
\end{thm3}
\begin{proof}
We know that $K_{n_1,\ldots,n_k}$ has $l$ edges, $|V(M(K_{n_1,\ldots,n_k}))|=l+n$. Let $V(M(K_{n_1,\ldots,n_k}))=\{v_1,\ldots,v_{l+n}\}$ and $n_0=0$. We assume that $v_1$ to $v_l$ represent the edges and for all $1 \leq i \leq k, v_{l+1+\sum_{j=0}^{i-1} n_j}$ to $v_{l+\sum_{j=0}^i n_j}$ represent the $i^{\rm th}$ partition vertices of $K_{n_1,\ldots,n_k}$. Let $r=\Delta,V_e=\{v_1,\ldots,v_l\}$ and $V_v=\{v_{l+1},\ldots,v_{l+n}\}$. It can be easily seen that $M(K_{n_1,\ldots,n_k})$ has a Vset-$d2r$ given by \\ $S_{d2r}=V_e \; \cup \; \{v_{l+1},v_{l+n_1+1},v_{l+(n_1+n_2)+1},v_{l+(n_1+n_2+n_3)+1},\ldots,v_{l+(n_1+\ldots+n_{k-1})+1}\}$. \\ Thus by lemma 1, $\chi_r(M(K_{n_1,\ldots,n_k})) \geq |S_{d2r}|= k+l$. We now define the coloring assignment $c \colon V(M(K_{n_1,\ldots,n_k}))$   $\to \{1,\ldots,k+l\}$ as follows: 
\begin{equation*}
c(v_i) =
\begin{cases}
i, & \text{if $1 \leq i \leq l.$}\\ 
l+p, & \text{otherwise, where $p$ is such that $1+\sum_{j=0}^{p-1} n_j \leq i-l \leq \sum_{j=0}^p n_j$.}
\end{cases}
\end{equation*}
In $c$ the first case uses $l$ colors and the remaining case uses $k$ new colors. We have $|c(V_e)|=|V_e|$ and $c(V_e) \cap c(V_v)= \emptyset$; for any two vertices $v_i,v_{i'} \in V_v$ if there exists no $p$ such that $l+1+\sum_{j=0}^{p-1} n_j \leq i,i' \leq l+\sum_{j=0}^i n_j$ then $c(v_i) \neq c(v_{i'})$; hence (C2) is satisfied at all the vertices. Taking $G=M(K_{n_1,\ldots,n_k})$ it can be seen that (C3) of lemma 2 is satisfied. By lemma 2, $\chi_\Delta (M(K_{n_1,\ldots,n_k})) \leq k+l$. Hence $\chi_\Delta (M(K_{n_1,\ldots,n_k})) = k+l$. \\
\end{proof}
\newtheorem{thm4}[thm1]{Proposition} 
\begin{thm4}
For $n \geq 4$, let $M(C_n)$ be the middle graph of $C_n$. Then \[\chi_r(M(C_n)) = \left\{  
\begin{array}{l l}
3, & \quad \mbox{if $r = 2$. \textsl{}}\\  
4, & \quad \mbox {if $ r =3$. \textsl{}}\\ \end{array} \right. \] 
\end{thm4}
\begin{proof}
Let $V(M(C_n))= \{v_1,\ldots,v_{2n}\}$. We assume that $v_1$ to $v_n$ and  $v_{n+1}$ to $v_{2n}$ represent the vertices and edges of $C_n$ respectively where for all $i$ ($1 \leq i \leq n$) $v_{n+i}$ is incident with both $v_{i}$ and $v_{i \; mod  \; n + 1}$. We have the following cases:  \\
\textbf{Case 1:} $r=2:$ Since $r < \Delta,\; \chi_r(M(C_n)) \geq 3$. We define the coloring assignment $c \colon V(M(C_n)) \to \{1,2,3 \}$ thus:\\  
For even $n$ 
\begin{equation*}
c(v_i) =
\begin{cases}
1, & \text{if $1 \leq i \leq n.$}\\
2, & \text{if $n+1 \leq i \leq 2n$ and $i:$ odd.}\\
3, & \text{otherwise.}
\end{cases}
\end{equation*}
For odd $n$  
\begin{equation*}
c(v_i) =
\begin{cases}
1, & \text{if $i=1$ or both $n+1 \leq i \leq 2n$ and $i:$ odd.}\\
2, & \text{if $i=2n$ or $2 \leq i \leq n-1$.}\\
3, & \text{otherwise.}
\end{cases}
\end{equation*}
In both the cases it can be verified that $c$ defines a conditional $(3,r)$-coloring of $M(C_n)$. Thus $\chi_r (M(C_n)) \leq 3$; hence  $\chi_r (M(C_n))=3$. \\
\textbf{Case 2:} $r=3:$ Since $r < \Delta, \; \chi_r(M(C_n)) \geq 4$. We define the coloring assignment $c \colon V(M(C_n)) \to \{1,2,3,4 \}$ as follows:
\begin{equation*}
c(v_i) =
\begin{cases}
1, & \text{if $n+1 \leq i \leq 2n$ and $(i-n):$ even.}\\
2, & \text{if $1 \leq i \leq n$ and $i:$ odd.}\\
3, & \text{if $i=n+1$ or both $4 \leq i \leq n$ and $i:$ even.}\\ 
4, & \text{otherwise.}
\end{cases}
\end{equation*} 
It can be verified that $c$ defines a conditional $(4,r)$-coloring of $M(C_n)$. Thus $\chi_r (M(C_n)) \leq 4$; hence  $\chi_r (M(C_n))=4$. 
\end{proof}
\newtheorem{thm5}[thm1]{Proposition} 
\begin{thm5}
Let $M(F_n)$ be the middle graph of $F_n$. Then \[\chi_r(M(F_n)) = \left\{  
\begin{array}{l l}
2n+1, & \quad \mbox{if $r \leq 2n$. \textsl{}}\\  
2n+2, & \quad \mbox {if $ r =2n+1$. \textsl{}}\\ 
2n+4, & \quad \mbox {if $ r =\Delta$. \textsl{}}\\ \end{array} \right. \]
\end{thm5}
\begin{proof}
From the definition we have $|V(M(F_n))|=5n+1$. Let $V(M(F_n))= \{v_1,\ldots,v_{5n+1}\}$. We assume that for all $i$ ($1 \leq i \leq n) \; v_{2i-1}$ and $v_{2i}$ represent the edges of $i^{\rm th}$ copy of $K_3$ incident with the center vertex, $v_{2(n+i)}$ and $v_{2(n+i)+1}$ represent the vertices of $i^{\rm th}$ copy of $K_3$ excluding the center vertex,$v_{4n+i+1}$ represents the edge of $i^{\rm th}$ copy of $K_3$ not incident with the center vertex and $v_{2n+1}$ represents the center vertex of $F_n$. It is clear that $\Delta(M(F_n))=2n+2$. We have the following cases\ : \\
\textbf{Case 1:} $r \leq 2n:$ Let $r'=2n$. Since $S=\{v_1,\ldots,v_{2n+1}\}$ is the maximum size clique of $M(F_n), \omega(M(F_n))=2n+1$. We know that $\chi_{r}(G) \geq \omega(G)$, taking $G=M(F_n)$ and $r=r'$, we get $\chi_{r'}(M(F_n))\geq 2n+1$. Now we define the coloring assignment $c \colon V(M(F_n)) \to \{1,\ldots,2n+1 \}$ as follows: 
\begin{equation*}
c(v_i) =
\begin{cases}
i, & \text{if $1 \leq i \leq 2n+1$.}\\
3, & \text{if $i=2n+2$.}\\
4, & \text{if $i=2n+3$.}\\ 
1, & \text{if $2n+4 \leq i \leq 4n+1$ and $(i-2n):$ even.}\\
2, & \text{if $2n+4 \leq i \leq 4n+1$ and $(i-2n):$ odd.}\\
2n+1, & \text{if $4n+2 \leq i \leq 5n+1$.}
\end{cases}
\end{equation*}
It is clear that the total number of colors used in $c$ is $2n+1$. Since $S$ is a clique, $|c(S)|=|S|=2n+1$ and $r'<2n+1$, for all $v \in S$ (C2) is satisfied at $v$. For all $i$  ($1 \leq i \leq n)$  $|c(\{v_{2i-1},v_{2i},v_{2(n+i)},v_{2(n+i)+1},v_{4n+i+1}\})|=|\{v_{2i-1},v_{2i},v_{2(n+i)},v_{2(n+i)+1},v_{4n+i+1}\}|$; therefore the remaining vertices also satisfy (C2). If we take $G$ to be $M(F_n)$ it follows (C3) of lemma 2 is satisfied. By lemma 2, $\chi_{r'} (M(F_n)) \leq 2n+1$; hence  $\chi_{r'} (M(F_n))= 2n+1$. From [3] we infer $\omega(G) \leq \chi_{r_1}(G) \leq \chi_{r_2}(G)$ if $r_1 \leq r_2$. Taking $G=M(F_n),r_1=r$ and $r_2=r'$, it follows that $\chi_r(M(F_n))= 2n+1$.\\ 
\textbf{Case 2:} $r=2n+1:$ Since $r < \Delta, \; \chi_r(M(F_n)) \geq r+1$. Now we define the coloring assignment $c' \colon V(M(F_n)) \to \{1,\ldots,2n+2 \}$ as follows: 
\begin{equation*}
c'(v_i) =
\begin{cases}
c(v_i)\, \; \text{as defined in case 1}, & \text{if $1 \leq i \leq 4n+1$}\; \text{and}\\
2n+2, & \text{otherwise.}
\end{cases}
\end{equation*}
In $c'$ the \textit{if}-case uses $2n+1$ and the \textit{otherwise}-case uses one new color. Since $|c'(S)|=|S|$ and $c'(S) \cap c'(\{v_{4n+2},\ldots,v_{5n+1}\})= \emptyset$, for all $v \in S$ (C2) is satisfied at $v$. By extending the argument similar to case 1, we can conclude that $c'$ defines a conditional $(2n+2,r)$-coloring of $M(F_n)$. Thus $\chi_{r} (M(F_n)) \leq 2n+2$; hence  $\chi_{r} (M(F_n))= 2n+2$. \\
\textbf{Case 3:} $r=\Delta:\,M(F_n)$ has a Vset-$d2r$ $S_{d2r}=\{v_1,\ldots,v_{2n+3},v_{4n+2}\}$; by lemma 1, $\chi_{r} (M(F_n)) \geq |S_{d2r}|= 2n+4$. We now define the coloring assignment $c \colon V(M(F_n)) \to \{1,\ldots,2n+4 \}$ as follows: 
\begin{equation*}
c(v_i) =
\begin{cases}
i, & \text{if $1 \leq i \leq 2n+3$.}\\
2n+3, & \text{if $2n+4 \leq i \leq 4n+1$ and $(i-2n):$ even.}\\
2n+4, & \text{if $i=4n+2$ or both $2n+4 \leq i \leq 4n+1$ and $(i-2n):$ odd.}\\
2n+2, & \text{if $4n+3 \leq i \leq 5n+1$.}
\end{cases}
\end{equation*}  
It is clear that in $c$ the total number of colors used is $2n+4$. For all $i$ ($1 \leq i \leq n)$ we have $|c(S) \cup c(\{v_{2(n+i)},v_{2(n+i)+1},v_{4n+i+1}\})|=|S \cup \{v_{2(n+i)},v_{2(n+i)+1},v_{4n+i+1}\}|$; therefore all the vertices satisfy (C2). With $G=M(F(n))$ we see that (C3) of lemma 2 is satisfied. Thus by lemma 2,  $\chi_r(M(F_n)) \leq 2n+4$. Hence  $\chi_r (M(F_n))= 2n+4$. 
\end{proof}
\newtheorem{thm6}[thm1]{Proposition}
\begin{thm6}
Let $M(K_{n_1,n_2})$ be the middle graph of $K_{n_1,n_2}$ and w.l.o.g. assume $n_1 \leq n_2$. Then 
\begin{equation*}
\chi_r (M(K_{n_1,n_2})) = 
\begin{cases}
n_2+1, & \text{if $\,r \leq n_2$.}\\
n_2+2, & \text{if $\,r = n_2+1$.} 
\end{cases}
\end{equation*}
\end{thm6}
\begin{proof}
We have $|V(M(K_{n_1,n_2}))|=n_1n_2+n$ where $n=n_1+n_2$. Let the vertex set $V(M(K_{n_1,n_2}))$ be $\{v_1,\ldots,v_{n+n_1n_2}\}$. We assume that $v_1$ to $v_{n_1}$ represent the vertices of the first partition, $v_{n_1+1}$ to $v_n$ represent the vertices of the second partition and $v_{n+1}$ to $v_{n+n_1n_2}$ represent the edges of $K_{n_1,n_2}$. We also assume that for all $i$ ($1 \leq i \leq n_1) \; v_{n+(i-1)n_2+1}$ to $v_{n+in_2}$ represent the edges incident at $v_i$ and for all $j$ ($1 \leq j \leq n_2) \; v_{n+(i-1)n_2+j}$ is incident with $v_i$ and $v_{n_1+j}$. It is clear that $\chi_\Delta (M(K_{n_1,n_2})) =n$. We have the following cases: \\ 
\textbf{Case 1:} $r \leq n_2:$ Let $r'=n_2$. Since $\{v_1,v_{n+1},\ldots,v_{n+n_2}\}$ is the maximum size clique of $M(K_{n_1,n_2})$, $\omega(M(K_{n_1,n_2}))=n_2+1$. We know, $\chi_{r}(G) \geq \omega(G)$; taking $G=M(K_{n_1,n_2})$ and $r=r'$, we get $\chi_{r'}(M(K_{n_1,n_2}))\geq n_2+1$. We define the coloring assignment $c \colon V(M(K_{n_1,n_2}))$ $\to \{1,\ldots,n_2+1 \}$ as follows: 
\begin{equation*}
c(v_i) =
\begin{cases}
n_2+1, & \text{if $1 \leq i \leq n$.}\\
1+(\lfloor(i-1-n)/n_2\rfloor+(i-n)) \mod {n_2}, & \text{otherwise.}
\end{cases}
\end{equation*}
In $c$ the \textit{if}-case uses one and the \textit{otherwise}-case uses $n_2$ new colors. For all $i$ ($1 \leq i \leq n_1) \; |c(N[v_i])|=|c(\{v_i,v_{n+(i-1)n_2+1},\ldots,v_{n+in_2}\})|=|\{v_i,v_{n+(i-1)n_2+1},\ldots,v_{n+in_2}\}|=|N[v_i]|$;  hence for all $v \in$ $V(M(K_{n_1,n_2})) \setminus \{v_{n_1+1},\ldots,v_n\}$ (C2) is satisfied at $v$. For all $j$ ($n_1+1 \leq j \leq n) \; |c(N(v_j))|=|c(\{v_{j+in_2}: \forall i: 1 \leq i \leq n_1 \})|=|\{v_{j+in_2}: \forall i: 1 \leq i \leq n_1 \}|=|N(v_j)|$; hence (C2) is satisfied at all $v_j$. Now (C3) of lemma 2 is satisfied if we set $G=M(K_{n_1,n_2})$. By lemma 2, $\chi_{r'} (M(K_{n_1,n_2})) \leq n_2+1$; hence  $\chi_{r'} (M(K_{n_1,n_2}))= n_2+1$. We know that $\omega(G) \leq \chi_{r_1}(G) \leq \chi_{r_2}(G)$ if $r_1 \leq r_2$. Taking $G=M(K_{n_1,n_2}),r_1=r$ and $r_2=r'$, it follows that $\chi_r(M(K_{n_1,n_2}))= n_2+1$.\\ 
\textbf{Case 2:} $r=n_2+1:$ Since $r < \Delta,\;$ $\chi_r(M(K_{n_1,n_2})) \geq r+1$. We now define the coloring assignment $c' \colon V(M(K_{n_1,n_2})) \to \{1,\ldots,n_2+2 \}$ as follows: 
\begin{equation*}
c'(v_i) =
\begin{cases}
n_2+2, & \text{if $1 \leq i \leq n_1$.}\\
c(v_i)\, \text{as defined in case 1}, & \text{otherwise.}
\end{cases} 
\end{equation*}
In $c'$ the \textit{if}-case uses one and the \textit{otherwise}-case uses $n_2+1$ new colors. $c'(\{v_1,\ldots,v_n\}) \cap c'(\{v_{n+1},\ldots,v_{n+n_1n_2}\})$ = $\emptyset,\,c'(\{v_1,\ldots,v_{n_1}\}) \cap c'(\{v_{n_1+1},\ldots,v_n\})=\emptyset$ and for all $i$ ($1 \leq i \leq n) \; |c'(N(v_i))|=|N(v_i)|$; hence (C2) is satisfied at all the vertices. By  setting $G=M(K_{n_1,n_2})$ we reason (C3) of lemma 2 is satisfied. By lemma 2,  $\chi_{r} (M(K_{n_1,n_2})) \leq n_2+2$. Hence  $\chi_r (M(K_{n_1,n_2})) = n_2+2$.
\end{proof}

\end{document}